\documentclass[runningheads]{llncs}
\usepackage[utf8]{inputenc}
\newtheorem{observation}{Observation}
\newtheorem{brule}{Branching Rule}

\usepackage{savesym,multirow,url,xspace,graphicx,hyperref,enumitem,color}
\usepackage[dvipsnames]{xcolor}
\usepackage[titletoc]{appendix}
\usepackage{subcaption}
\usepackage{dsfont}
\usepackage{multicol}

\usepackage{wrapfig}

\usepackage{etoolbox}
\usepackage[utf8]{inputenc} 
\usepackage[T1]{fontenc}    
\usepackage{hyperref}       
\usepackage{url}            
\usepackage{booktabs}       
\usepackage{amsfonts}       
\usepackage{nicefrac}       
\usepackage{microtype}      
\usepackage{tikz}
\usepackage{pgfplots}
\usepackage{verbatim}
\usepackage{multirow}
\usepackage{mathtools}
\usepackage{comment}

\newtoggle{showcomments}
\settoggle{showcomments}{true}

\iftoggle{showcomments}{
\newcommand{\akash}[1]{{\textcolor{purple}{[{\bf Akash:} #1]}}}
\newcommand{\adish}[1]{{\textcolor{red}{[{\bf Adish:} #1]}}}
\newcommand{\yuxin}[1]{{\textcolor{ForestGreen}{[{\bf Yuxin:} #1]}}}
\newcommand{\hank}[1]{{\textcolor{NavyBlue}{[{\bf Hank:} #1]}}}
\newcommand{\annotate}[1]{{\textcolor{blue}{[{\bf Note:} #1]}}}
}
{
\newcommand{\akash}[1]{}
\newcommand{\adish}[1]{}
\newcommand{\yuxin}[1]{}
\newcommand{\hank}[1]{}
\newcommand{\annotate}[1]{}
}

\makeatletter

\newcommand{\Rmnum}[1]{\expandafter\@slowromancap\romannumeral #1@}
\makeatother
\usepackage{caption}

\newcommand{\secref}[1]{\S\ref{#1}}
\newcommand{\appref}[1]{Appendix \ref{#1}}

\newcommand{\thmref}[1]{Theorem~\ref{#1}}

\newcommand{\lemref}[1]{Lemma~\ref{#1}}
\newcommand{\obref}[1]{Observation~\ref{#1}}

\newcommand{\paren} [1] {\ensuremath{ \left( {#1} \right) }}

\newcommand{\bracket}[1]{\left[#1\right]}

\newcommand{\curlybracket}[1]{\ensuremath{\left\{#1\right\}}}

\newcommand{\bigO}[1]{\ensuremath{\mathcal{O}\paren{#1}}}

\newcommand{\cP}{{\mathcal{P}}}

\newcommand{\bc}{{\mathbf{c}}}

\renewcommand{\tt}[1]{\textit{#1}}

\def\BState{\State\hskip-\ALG@thistlm}

\usepackage{graphicx}
\begin{document}
\title{Deletion to Induced Matching}
%
%
\author{Akash Kumar\inst{1}\orcidID{0000-0002-0720-9320} \and
Mithilesh Kumar\inst{2}\orcidID{0000-0003-0077-2118}}
\authorrunning{A. Kumar et al.}
%
\institute{Max Planck Institute for Software Systems, Campus E1 5,
D-66123 Saarbr\"{u}cken, Germany\\
\email{akumar@mpi-sws.org} \and
Simula@UiB, Merkantilen, Thormøhlens gate 53D, N-5006 Bergen, Norway\\
\email{thesixthprime@gmail.com}}
\maketitle              
\begin{abstract}
In the \textsc{Deletion to induced matching} problem, we are given a graph $G$ on $n$ vertices, $m$ edges and a non-negative integer $k$ and asks whether there exists a set of vertices $S \subseteq V(G) $ such that $|S|\le k$ and the size of any connected component in $G-S$ is \textit{exactly} 2. In this paper, we provide a fixed-parameter tractable (FPT) \(O^*(1.748^{k})\) running time and polynomial space algorithm for the  \textsc{Deletion to induced matching} problem using \textit{branch}-and-\textit{reduce} strategy and path decomposition. We also extend our work to the exact-exponential version of the problem.

\end{abstract}
\vspace{-1mm} 
\section{Introduction}
Hardness of a computation problem often depends on the class of the underlying graph. Given a graph, it becomes natural to check \emph{how far} is the graph from a specific graph class. One way to quantify this \emph{distance} is in terms of number of vertices that need to be deleted from the given graph such that the resultant graph belongs the desired graph class. \textsc{Deletion to Induced matching} is one such problem. Before we define this problem, let's take a brief detour to some of the well studied problems.

In the classic {\sc Vertex Cover} problem, the input is a graph $G$ and integer $k$, and the task is to determine whether there exists a vertex set $S$ of size at most $k$ such that every edge in $G$ has at least one endpoint in $S$. Such a set is called a {\em vertex cover} of the input graph $G$. An equivalent definition of a vertex cover is that every connected component of $G - S$ has at most $1$ vertex. This view of the {\sc Vertex Cover} problem gives rise to a natural generalization: can we delete at most $k$ vertices from $G$ such that every connected component in the resulting graph has at most $\ell$ vertices? $\mathrm{Kumar\, et\, al.}$~\cite{kumar_et_al:LIPIcs:2017:6934} studied this generalization as $\ell$-\textsf{COC} (\textsf{$\ell$-Component Order Connectivity}).
In this work, we would study a special case of this generalization where $\ell$ is \tt{exactly} 2. Formally, we consider the following problem, called \textsf{Deletion to Induced Matching} (\textsf{IND}).

\smallskip
\noindent
\fbox{\parbox{\textwidth-\fboxsep}{
\textsf{Deletion to Induced Matching} (\textsf{IND})\\
\textbf{Input:} A graph $G$ on $n$ vertices and $m$ edges, and a positive integer $k$.\\
\textbf{Task:} determine whether there exists a set $S\subseteq V(G)$ such that $|S|\leq k$ and the maximum size of a component in $G-S$ is \emph{exactly} 2.
}}
\smallskip

\noindent
Exact 2-\textsf{COC} is also known in the literature as \tt{induced matching} which is a well-studied problem. A maximum induced matching problem where the task is to find an induced matching of maximum size, has been shown to be solvable in polytime for various graph classes, e.g. trees \cite{GOLUMBIC2000157}, chordal graphs \cite{cam}, circular-arc graphs \cite{irred} and interval graphs \cite{GOLUMBIC2000157}. From the work of Stockmeyer and Vazirani~\cite{Stockmeyer1982NPCompletenessOS}, it is evident that \textsf{IND} is NP-complete. This motivates the study of \textsf{IND} within paradigms for coping with NP-hardness, such as approximation algorithms~\cite{approx_book}, exact exponential time algorithms~\cite{exact_book}, parameterized algorithms~\cite{pc_book,DowneyF13book} and kernelization~\cite{kernel_survey_kratsch}. In this work we focus on \textsf{IND} from the perspective of parameterized complexity and exact exponential algorithms. As our main result, we provide an algorithm that given an instance $(G, k)$ of \textsf{IND} such that degree of any vertex is at most 3, runs in polynomial time to output a path decomposition such that the path width is bounded by $\bigO{k}$. We provide an application of branching technique to convert an arbitrary instance $(G, k)$ of \textsf{IND} to $(G', k')$ such that degree of any vertex in $G'$ is at most 3.   

\medskip
\noindent
{\bf Related Work.} If the component size is 1, then the problem converts to finding a \textsf{VERTEX COVER}, which is extremely well studied from the perspective of approximation algorithms~\cite{approx_book,dinur2005hardness}, exact exponential time algorithms~\cite{FominGK09,Robson86,XiaoN13}, parameterized algorithms~\cite{pc_book,ChenKX10} and kernelization~\cite{ChenKJ01,NemhauserT74}. The relaxed version of \textsf{IND} where component sizes are bounded by 2 (also 2-\textsf{COC}~\cite{kumar_et_al:LIPIcs:2017:6934}) is also well studied, and has been considered under several different names. The problem, or rather the dual problem of finding a largest possible set $S$ that induces a subgraph in which every connected component has order at most 2, was first defined by Yannakakis~\cite{Yannakakis81a} under the name Dissociation Set. The problem has attracted attention in exact exponential time algorithms~\cite{KardosKS11,XiaoK15}, the fastest currently known algorithm~\cite{XiaoK15} has running time $O(1.3659^n)$. $2$-\textsf{COC} has also been studied from the perspective of parameterized algorithms~\cite{ChangCHRS16,Tu15} (under the name {\sc Vertex Cover} $P_3$) as well as approximation algorithms~\cite{Tu}. The fastest known parameterized algorithm, due to Tsur et al.~\cite{ChangCHRS16} has running time $1.713^kn^{O(1)}$, while the best approximation algorithm, due to Tu and Zhou~\cite{Tu} has factor $2$. 

\textsf{IND} has been studied in the literature as maximum induced matching (\textsf{MIM}) both in parameterized and exact-exponential settings. Xiao et al.~\cite{AIMpara,XIAO2020} claims a fixed pararmeter tractable algorithm of running time $1.748^kn^{O(1)}$ in the polynomial space. The exact exponential version of \textsf{IND} which could be stated as finding a maximum induced regular graph of degree 1, has been first studied by Saket el al.~\cite{regularinducedsub}. They have shown an improved exponential time algorithm of running time $O^{*}(1.4786^n)$. The running time was further improved to $O^{*}(1.4231^n)$ time and polynomial space by Xiao et al.~\cite{AIMexact}. Basavaraju et al.~\cite{basava} showed that all maximal induced matchings in a triangle-free graph can be listed in $O^{*}(1.4423^n)$ time. In our work, we achieve a running time of $1.748^kn^{O(1)}$ in the polynomial space for the parameterized problem \textsf{IND} using novel techniques from branching and path decomposition. Our solution is intuitive and uses a much simpler version of branching as has been used in previous works. These techniques could be leveraged to provide a simple $O^*(1.5009^{n})$-time exact algorithm for the exact exponential version of the problem, called \textsf{EXTEND} in polynomial space.
Using monotone local search for exact problems as shown in the seminal work of Saket et al.~\cite{monotone}, an $O^*(1.427^{n+o(n)})$ running time algorithm is also proposed.

\medskip
\noindent
{\bf Our Method.} The key to our solution is reducing an \textsf{IND} instance $(G,k)$ to $(G',k')$ using branch-and-reduce strategy such that the maximum degree of any vertex $u \in G'$ is 3. We achieve this by analysing certain properties of a solution set $S \subseteq V(G)$. Thus, we only apply the branch-and-reduce strategy on vertices with degree more than or equal to 4. Consider a search tree $T^*$ and the subtree rooted at some node with \textsf{IND} instance $(G',k')$. Now, for any $u \in V(G)$, either $u \in S$ or $u \notin S$. In the former case, the algorithm creates a child with instance $(G'[\paren{V(G')\setminus u}], k'-1)$. In the later case, we study the $e \in E(G')$ such that $e := uv$ for some $v \in N(u)$ (neighbours of $u$ in $G'$). Our branching rules cover all the possible cases of $\{u,v\} \not\subset S$ where $v \in N(u)$.

Now, the motivation behind reducing the \textsf{IND} instances to maximum degree 3 instances, call $(G',k')$, is because we can efficiently solve those instances using ideas from path decomposition. We are able to show that there exists a path decomposition of $G'$ with pathwidth having constant dependency on $k'$. Not only that we provide an algorithm that runs in polynomial time to construct one.

We note that Fomin and H\o{}ie~\cite{fominho} proved a tight bound on a path decomposition of a graph $G$ with maximum degree at most three. They show that for any $\epsilon > 0$, \textsf{pw}$(G)\leq  (\frac{1}{6} + \epsilon)|V (G)|$ (pathwidth of the graph, cf \secref{sec:prelim}). We first show that in the case of the special degree 3 graphs, we can bound the number of vertices of degree exactly 3 if $(G,k)$ is a \textsf{YES}-instance, and vice versa. Now the key is to show a path decomposition of the instance of width bounded by $O(k)$. The main insight is in using the decomposition given by Fomin and H\o{}ie~\cite{fominho}, call it $\mathcal{P}'$, and then applying color coding of edges and vertices to appropriately group the bags of the decomposition $\mathcal{P}'$ (cf \secref{sec:prelim}) to construct the desired path decomposition $\mathcal{P}$.

Now, the only question left to be answered is how to construct an efficient solution for the instance $(G,k)$? 
We provide a dynamic programming algorithm to solve an \textsf{IND} instance $(G,k)$ if a path decomposition of $G$ is given. We define a coloring map of any bag $X_t \in \mathcal{P}$ i.e. the path decomposition $\mathcal{P}$ of $G$ in the following manner: $f:X_t\to \{0,1,2\}$ assigning three different colors to vertices of the bag. For the path decomposition $\cP$ is represented as $(\mathcal{P}, \{X_t\}_{t\in V(\mathcal{P})})$. The idea is to dynamically find partial solution $S_t \subseteq G[V_t]$ where $V_t:=\bigcup_{i=1}^{i=t}X_i$ for any $t \in V(\cP)$. Thus, the coloring set $\{0,1,2\}$ induces a canonical meaning with the color $0$ is when the vertex is in the partial solution, the color $1$ is assigned when the vertex is not in the partial solution but to be paired later in the DP, and the color $2$ is assigned when the vertex doesn't belong to the partial solution  but has degree one in $G_t - S_t$. Since, we have to ensure the minimum property of $S_t$ we further define a cost function $\bc[t,f]$ the minimum size of a set $S_t\subseteq V_t$ such that the following two properties hold:
\begin{enumerate}
\looseness-1
\item $S_t\cap X_t=f^{-1}(0)$, i.e. the set of vertices of $X_t$ that belong to the partial solution.
\item Degree of every vertex in $G_t-S_t$ is at most 1.
\end{enumerate}
\looseness-1
We show the $\bc[\cdot,\cdot]$ could be dynamically constructed. Since the maximum possible colorings $f$ of a bag $X_t$ is bounded by $3^{|X_t|}$ the running time is further bounded.

The exact-exponential algorithm of the problem which involves finding the minimum set $S \subseteq V(G)$ such that $G[V-S]$ contains only components of size exactly 2, follows similar branching and reducing strategy as devised for the parameterized problem $\textsf{IND}$. We use the key results proven for the parameterized to analyse this case and provide a worst-case solution.


\medskip
\noindent
{\bf Overview of the paper.} 
In \secref{sec:prelim} we recall basic definitions and set up notations.
We provide a self-contained section for the FPT algorithm and the corresponding necessary results in \secref{sec:ind}.
In \secref{sec:extend}, we state the exact-exponential problem and provide the exact-exponential algorithm. 



\looseness -1
\section{Preliminaries}\label{sec:prelim}
Let $\mathbb{N}$ denote the set of positive integers $\{0,1,2,\dots\}$. For any non-zero $t\in \mathbb{N}$, $[t]:=\{1,2,\dots,t\}$.
For a set $\{v\}$ containing a single element, we simply write $v$. A vertex $u\in V(G)$ is said to be incident on an edge $e\in E(G)$ if $u$ is one of the endpoints of $e$. A pair of edges $e,e'\in E(G)$ are said to be adjacent if there is a vertex $u\in V(G)$ such that $u$ is incident on both $e$ and $e'$. For any vertex $u\in V(G)$, by $N(u)$ we denote the set of neighbors of $u$ i.e. $N(u):=\{v\in V(G)\mid uv\in E(G)\}$. We use the notation $N[u]$ to denote the union $u \cup N(u)$, where $d(u)$ denotes the degree of the vertex $u$.
For any subgraph $X\subseteq G$, by $N(X)$ we denote the set of neighbors of vertices in $X$ outside $X$, i.e. $N(X):=\paren{\bigcup_{u\in X}N(u)}\setminus X$. 
An induced subgraph on $X\subseteq V(G)$ is denoted by $G[X]$.

A \emph{path} $P$ is a graph, denoted by a sequence of vertices $v_1v_2\dots v_t$ such that for any $i,j\in [t]$, $v_iv_j\in E(P)$ if and only if $|i-j|=1$. 
The \emph{length} of a path is the number of edges in the path.

\medskip
\noindent
{\bf Path Decomposition.}\cite{parameterizedbook} A path decomposition of a graph $G$ is a sequence $\cP = (X_1, X_2,\cdots,X_r)$ of bags where $X_i \subseteq V(G)$ for each $i \in \curlybracket{1,2,\cdots,r}$, such that the following conditions hold:
\begin{itemize}
    \item[$\bullet$] $\bigcup_{i =1}^r X_i= V(G)$. In other words, every vertex of $V(G)$ is in at the least one of the bag. 
    \item[$\bullet$] For every $uv \in E(G)$, there is an $\ell \in \curlybracket{1,2,\cdots,r}$ such that the bag $X_{\ell}$ contains both $u$ and $v$.
    \item[$\bullet$] For every $u \in V(G)$, if $u \in X_i \cap X_k$ for some $i \le k$, then $u \in X_j$ also for each $j$ such that $i\le j \le k$. In other words, the indices of the bags containing $u$ form an interval in $\curlybracket{1,2,\cdots,r}$.
\end{itemize}
The width of a path decomposition $(X_1, X_2,\cdots,X_r)$ is $\max_{1\le i \le r} |X_i| - 1$. The pathwidth of a graph $G$, denoted by ${\bf pw} (G)$, is the minimum possible width of a path decomposition of $G$. The reason for subtracting 1 in the definition of the width of the path decomposition is to ensure that the path width of a path with at least one edge is 1, not 2.

\medskip
\noindent
{\bf Fixed Parameter Tractability.} A {\em parameterized problem} $\Pi$ is a subset of $\Sigma^* \times \mathbb{N}$. A parameterized problem $\Pi$ is said to be \emph{fixed parameter tractable}(\textsc{FPT}) if there exists an algorithm that takes as input an instance $(I, k)$ and decides whether $(I, k) \in \Pi$ in time $f(k)\cdot n^c$, where $n$ is the length of the string $I$, $f(k)$ is a computable function depending only on $k$ and $c$ is a constant independent of $n$ and $k$. 

A \emph{data reduction rule}, or simply, reduction rule, for a parameterized problem $Q$ is a function $\phi:\Sigma^*\times\mathbb{N}\to \Sigma^*\times\mathbb{N}$ that maps an instance $(I,k)$ of $Q$ to an equivalent instance $(I',k')$ of $Q$ such that $\phi$ is computable in time polynomial in $|I|$ and $k$. We say that two instances of $Q$ are \emph{equivalent} if $(I,k)\in Q$ if and only if $(I',k')\in Q$; this property of the reduction rule $\phi$, that it translates an instance to an equivalent one, is referred as the \emph{safeness} of the reduction rule.

A fixed-parameter algorithm based on branch-and-reduce strategy consists of a collection of reduction rules and branching rules. 
The branching rules are used to recursively solve the smaller instances of the problem with smaller parameter. We analyze each branching rule and use the worst-case time complexity over all branching rules as an upper bound of the running time. 
We represent the execution of a branching algorithm via search tree. The root of a search tree represents the input of the problem, every child of the root represents a smaller instance reached by applying a branching rule associated with the instance of the root. One can recursively assign a child to a node in the search tree when applying a branching rule. Notice that we do not assign a child to a node when applying a reduction rule. The running time of a branching algorithm is usually measured by the maximum number of leaves in its corresponding search tree.
Let $b$ be any branching rule. When rule $b$ is applied, the current instance $(G,k)$ is branched into $s\ge 2$ instances $(G_i,k_i)$ where $|G_i| \le |G|$ and $k_i = k - t_i$. Notice that fixed-parameter algorithms return “No” when the parameter $k \le 0$. We call $\textbf{b} = (t_1,t_2,\cdots, t_s)$ the branching vector of branching rule $b$. This can be formulated in a linear recurrence:
$T(k) \le T(k-t_1)+T(k-t_2)+\cdots+T(k-t_s)$, where $T(k)$ is the number of leaves in the search tree depending on the parameter $k$. The running time of the branching algorithm using only branching rule $b$ is $O(\textit{poly}(n)\cdot T(k)) = O^*(c^k)$ where $c$ is the unique positive real root of $x^k - x^{k-t_1}- x^{k-t_2}-\cdots-x^{k-t_s} = 0$~\cite{exact_book}. The number $c$ is called the branching number of the branching vector $(t_1,t_2,\cdots, t_s)$.

\section{Faster \textsf{FPT} algorithm for \textsf{Induced Matching}}\label{sec:ind}
In this section, we would present the fixed parameter tractable algorithm for \textsf{Deletion To Induced Matching}. First, we construct a set of branching rules for problem instances with at the least one vertex with degree $\ge 4$. The branching rules are based on a novel observation. Using these rules, we get a problem instance $(G',k')$ from $(G,k)$ such that any vertex $u \in V(G')$ has degree bounded by 3. We provide an efficient solution for the instance $(G',k')$ using ideas from path decomposition.
We would provide a dynamic programming algorithm for instances with maximim degree 3. Overall, the running time of the algorithm is $O^*(1.748^k)$ and does computation in polynomial space. In contrast to the work of Xiao et al.~\cite{AIMpara}, our solution uses many fewer branching rules and uses novel techniques from path decomposition of a graph.\\

We denote by $S$ a potential solution of size at most $k$. Now we proceed to write the reduction rules and branching rules. Note that while stating a reduction rule or a branch rule we assume that previous rules are not applicable. Furthermore, each rule changes the instance from $(G,k)$ to $(G',k')$ where
$|V(G')|<|V(G)|$ and $k'\leq k$, but we use the same symbols $(G,k)$ to represent the modified instance. We note the following key observation:   
\begin{observation}\label{contained}
If $\exists u,v\in V(G)$ such that $N[v]\setminus u\subseteq N(u)$, then there exists a solution $S$ such that either $u\in S$ or $\{u,v\}$ $\not\in S$.
\end{observation}

Indeed if $S$ is a solution such that $u\notin S$ and $v\in S$ such that $\exists w \in N(u)\cap N(v)$ and $w \notin S$, then $S':=\paren{S\setminus v}\cup u$ is also a solution where $v \not\in S$. If $\exists x \not\in N(v)$ such that $x \not\in S$, then note that all the neighbours of $x$ outside that of $u$ has to be deleted. Thus, $S':=\paren{S\setminus v}\cup x$ where $x \in S$ and $v \not\in S$. Hence, even for $u\notin S$, we get a solution such that either $u\in S$ or $\{u,v\}$ $\not\subseteq S$.
\\
\\
\obref{contained} suggests the following branching rules where we assume that the degree of $u$ is at the least 4 i.e. $d(u) \geq 4$, where the branching subtree is rooted at $u$ :
\begin{brule}\label{one}
If  $\exists v \in V(G)$ such that $N[v] \subseteq N[u]$. Then, make nodes in the branch tree for the following cases: $u \in S$ or $u \notin S$. In the second case, we pair the vertex $u$ with $v$ and delete at least $d(u)-1$ many vertices. The recurrence relation is $T(k)\leq T(k-1)+T(k-d+1)$ whose solution is bounded by $T(k)\leq 1.4656^k$.
\end{brule}
\noindent
From now on we assume that for every vertex $u$ of degree at least $4$, we have that for every vertex $v\in N(u)$, $|N(v)\cap \overline{N[u]}|\geq 1$. If $u\notin S$, then one neighbor $v\in N(u)$ does not belong to $S$. In that case, $N(v)\cap \overline{N[u]}\subseteq S$. This is the basis for the following branching rule:
\begin{brule}\label{two}
Let $u$ be a vertex of degree $d\geq 4$ such that $\forall v\in N(u)$\,  $|N(v)\cap \overline{N[u]}|\geq 1$.
Create $d+1$ branch nodes: one for the case when $u\in S$ and one for each vertex $v\in N(u)$ such that 
$\paren{N(u)\cup N(v)}\setminus \{u,v\}\subseteq S$ where at least $d$ vertices are deleted.  Then, the recurrence 
relation is $T(k)\leq T(k-1)+d\cdot T(k-d)$ whose solution is bounded by $T(k)\leq 1.748^k$.
\end{brule}
\noindent
Using the above Branching Rules \ref{one}-\ref{two}, we can reduce any \textsf{IND} instance $(G,k)$ to $(G',k')$ such that $G'$ has maximum degree 3 for any vertex $u \in V(G')$. Now, we would provide the construction of a path decomposition for $(G',k')$. First, we state the following result by Fomin and H\o{}ie~\cite{fominho}:
\begin{theorem}[Fomin and H\o{}ie]\label{fedor}
For any $\epsilon > 0$, there exists an integer $n_{\epsilon}$ such that for every graph $G$ with maximum vertex degree at most three and with $|V(G)| > n_{\epsilon}$ , \textsf{pw}$(G)\leq  (\frac{1}{6} + \epsilon)|V (G)|$. Furthermore, such a decomposition can be obtained in polynomial time.
\end{theorem}
Now, we would prove a lemma which ascertains a key property on the number of vertices with degree 3 for an \textsf{IND} problem instance $(G,k)$ to be a \textsf{YES}-instance.

\begin{lemma}\label{bound}
Let $G$ be a graph of maximum degree 3. Then $(G,k)$ is a \textsf{yes}-instance if and only if there are at 
most $2.5k$ vertices of degree 3 in $G$.
\end{lemma}
\begin{proof}
Let $(G,k)$ be a \textsf{yes}-instance of \textsf{IND} and $S$ be a solution of size at most $k$. Since the degree of any vertex is at most 3, the number of edges between $S$ and $V\setminus S$ is at most $3k$. Now, the degree of any vertex in $G-S$ is  1. Hence, for every vertex of degree 3 in $V\setminus S$, the number of edges to $S$ is 2. Hence, the maximum number of degree 3 vertices is bounded by $k+\frac{3}{2}k=2.5k$.

If number of degree 3 vertices is more than $2.5k$, then for any set $S$ of size $k$, $G-S$ contains a vertex of degree greater than 1 and hence $(G,k)$ is a \textsf{no}-instance.
\end{proof}
Using \thmref{fedor} and \lemref{bound}, we can bound the path width of the path decomposition of an \textsf{IND} instance $(G,k)$ to $\frac{2.5k}{6}+2$. We state the result here with detailed proof in \appref{appendix: pathwidth}.
\begin{lemma}\label{width}
There exists a polynomial time algorithm that given an \textsf{IND} instance $(G,k)$ obtains a path decomposition of $G$ of width at most $\frac{2.5k}{6}+2$. 
\end{lemma}
Before, we show our algorithm for solving an \textsf{IND} instance with every vertex $u \in G$ having maximum degree 3, we would give an algorithm to solve an \textsf{IND} instance in 
$O^*(3^p)$ time where $p$ is the pathwidth of a path decomposition for $G$:
\begin{theorem}\label{dp}
There exists an algorithm that given a path decomposition of $G$ solves \textsf{IND} in $O^*(3^p)$ time where $p$ is the width of the decomposition.
\end{theorem}
\begin{proof}
We would assume that we are given a nice path decomposition $(\mathcal{P}, \{X_t\}_{t\in V(\mathcal{P})})$ of 
width $p$ of $G$. We define $V_t:=\bigcup_{i=1}^{i=t}X_i$ and $G_t:= G[V_t]$ for any $t \in V(\mathcal{P})$. We define a coloring of a bag $X_t$ as a
mapping $f:X_t\to \{0,1,2\}$ assigning three different colors to vertices of the bag, which is described extensively as follows:
\begin{itemize}
    \item[$\bullet$] If $f(v)=0$, then $v$ must be contained in the partial solution in $G_t$.
    \item[$\bullet$] If $f(v)=1$, then $v$ does not belong to the partial solution, but it is isolated in $G_t-S_t$ where $S_t$ is the partial solution in $G_t$. These are the vertices which are to be paired later in the DP (dynamic programming of the algorithm).
    \item[$\bullet$] If $f(v)=2$, then $v$ does not belong to the partial solution, but it has degree one in $G_t-S_t$. Note that $d(v)\leq 1$ in $G[X_t-f^{-1}(0)]$.
\end{itemize}
For a subset $X\subseteq V(G)$, consider a coloring $f:X\to \{0,1,2\}$. For a vertex $v\in V(G)$ and a color 
$\alpha\in \{0,1,2\}$ we define a new coloring $f_{v\rightarrow \alpha}:X\cup \{v\}\to \{0,1,2\}$ as follows:
$$f_{v\rightarrow \alpha}(x)=
\begin{cases}
f(x) & \mbox{ when } x\neq v,\\
\alpha & \mbox{ when } x=v.
\end{cases}
$$
Similarly, we define another coloring $f_{v\leftarrow \alpha}:X\to \{0,1,2\}$ as follows:\\
$$f_{v\leftarrow \alpha}(x)=
\begin{cases}
f(x) & \mbox{ when } x\neq v,\\
\alpha & \mbox{ when } x=v.
\end{cases}
$$
For a coloring $f$ of $X$ and $Y\subseteq X$, we use $f|_Y$ to denote the restriction of $f$ to $Y$.
For a coloring $f$ of $X_t$, we denote by $\bc[t,f]$ the minimum size of a set $S_t\subseteq V_t$ such that the following two properties hold:
\begin{enumerate}
\item $S_t\cap X_t=f^{-1}(0)$, i.e. the set of vertices of $X_t$ that belong to the partial solution.
\item Degree of every vertex in $G_t-S_t$ is at most 1.
\end{enumerate}
Since the decomposition is \tt{nice}, we specify the values of the recursive function $\bc\bracket{\cdot,\cdot}$ to various node types as shown below:\\

\textbf{Leaf node}: For a leaf node $t$ we have that $X_t=\emptyset$. Hence there is only one, empty coloring,
and we have $\bc[t,\emptyset]=0$.\\

\textbf{Introduce node}: Let $t$ be an introduce node with a child $t'$ such that $X_t=X_{t'}\cup \{v\}$ for
some $v\notin X_{t'}$. We write the recursive formulae for various cases. Note that we put infinity as a value
for $\bc[t,f]$ whenever a feasible solution is not possible.\\
$$\bc[t,f]=
\begin{cases}
1+\bc[t',f|_{X_{t'}}] & \mbox{ if } f(v)=0\\
\bc[t',f|_{X_{t'}}] & \mbox{ if } f(v)=1 \mbox{ and } \not\exists u\in X_t \mbox{ s.t. } \\
&\mbox{                                       }uv\in E(G) \mbox{ and } f(u) \neq 0 \\ 
\bc[t',f_{w\leftarrow 1}|_{X_{t'}\setminus w}] & \mbox{ if } f(v)=2, \exists w \mbox{ s.t. } f(w) = 1 \mbox{ and } wv \in E(G)\\ 
& \mbox{ and }\not\exists u \mbox{ s.t. } f(u)=2 \mbox{ and } uv\in E(G) \\
\infty & \mbox{ otherwise}.
\end{cases}
$$ 
\\
Note that when $vw\in E(G)$ and $f(w)=1$ in $X_{t'}$, then at the \textbf{introduce node} $X_t$, we have $f(w)=2$.\\

\textbf{Forget node}: Let $t$ be a forget node with a child $t'$ such that $X_t=X_{t'}\setminus \{v\}$ for
some $v\in X_{t'}$. Since multiple colorings of $X_{t'}$ can lead to the same coloring of $X_t$, it suffices
to keep the one that leads to minimum size solution. 
\begin{center}
$\bc[t,f]=\min \curlybracket{\bc[t', f_{v\rightarrow 0}],\,\bc[t', f_{v\rightarrow 1}],\,\bc[t', f_{v\rightarrow 2}]}$
\end{center}
If an introduced vertex is colored 2 then, the coloring is  valid unless there is a vertex $w \in X_{t'}$ such that the restriction \(f|_{X_{t'}}\)
satisfies $f|_{X_{t'}}(w)=1$. 

Since there are at most $3^{|X_t|}$ number of colorings $f$ for any bag $X_t$, the time to process any node
is at most $3^{p+1}$. Hence, \textsf{IND} can be solved in $O^*(3^p)$ time.
\end{proof}
\begin{lemma}\label{last}
\textsf{IND} on graphs of maximum degree 3 can be solved in $O^*(1.581^k)$ time.
\end{lemma}
\begin{proof}
Run the algorithm of \lemref{width} to get a nice path decomposition of width at most $\frac{2.5k}{6}+2$ and 
then use the algorithm of \thmref{dp}. Hence, the running time of the algorithm is bounded by 
$O^*(3^{\frac{2.5k}{6}+2})=O^*(1.581^k)$. 
\end{proof}
We have provided a faster $\textsf{FPT}$ solution for problem instances with maximum degree 3. Using the branching rules \ref{one}-\ref{two} and \lemref{last}, we would state and proof the main claim of the section in which we show a fixed-parameter tractable algorithm of running time $O^*(1.748^k)$ for any \textsf{IND} instance $(G,k)$. In the following, we state the claim below and then we complete the proof. 
\begin{theorem}[Main Theorem]\label{ind}
\textsf{IND} can be solved in $O^*(1.748^k)$-time and polynomial space.
\end{theorem}
\begin{proof}
Algorithm solves \textsf{IND} in two phases. In phase one, it applies some branching rules to 
eliminate at the least degree 4 vertices.  In phase two, it solves \textsf{IND} on graphs of maximum degree 3.

After exhaustively applying  Branching Rules 
\ref{one} and \ref{two} on an \textsf{IND} instance $(G,k)$, we can assume that for the reduced instance $(G',k')$ the maximum degree of a vertex in the 
reduced graph $G'$ is 3. 
At this point we run the algorithm of \lemref{last} to solve \textsf{IND}
optimally. Hence, if $(G,k)$ is
a \textsf{YES}-instance of \textsf{IND}, then in at least one branch, we get a solution of size at most $k$.

Now we proceed to the running time analysis of the algorithm. First, we write the recurrence relations for each of 
the branch rules we have used so far. 
\begin{align*}
    \textnormal{Branch rule}~\ref{one}\quad& T(k)\leq T(k-1)+T(k-d+1) & T(k)\leq 1.465^k\\
    \textnormal{Branch rule}~\ref{two}\quad& T(k)\leq T(k-1)+d\cdot T(k-d) & T(k)\leq 1.748^k
\end{align*}
Consider the branch tree at the end of the branch phase. Let $s$ be the parameter with which the algorithm of 
Lemma \ref{last} is called. Clearly, $0< s\leq k$. It is easily shown by induction on $k$ that the number of 
leaves with parameter $s$ is bounded by the worst case branching in the branch tree which is $1.748^{k-s}$: 
indeed the base  case $k=1$ is trivial. Consider that the statement is true for any value less than $k$. 
Let $T$ represent the branch tree. Then consider all the nodes with parameter $s$ in $T$. If there is a path on
which there is no such node then, delete the subtree rooted at the node that has an edge to the path that leads 
to a node with parameter $s$. Now delete all subtrees rooted at nodes with parameter $s$ except the root. In this
truncated tree defined as $T^*$, leaves have parameter $s$. Now, change $k$ to $k'=k-s$. Due to this parameter change, all 
leaves are with parameter 1. Since, $k'<k$, by induction, the number of nodes with parameter value 1 is bounded
by $1.748^{k'}=1.748^{k-s}$. Hence, the total running time of the algorithm is bounded as follows \\
\begin{equation*}
    \sum\limits_{s=0}^{s=k}1.748^{k-s}\times 1.581^{s}\times n^{O(1)}\leq \sum\limits_{s=0}^{s=k}1.748^s\times 1.748^{k-s}\times n^{O(1)}\leq O^*(1.748^k)
\end{equation*}

Space used by the algorithm can be reduced to a polynomial by defining an ID for each node in the search tree that encodes what branching order to follow and the current branch. One way to do this is to write the vertices to be branched on in a lexicographic order and then write the current vertex being branched on. Since the depth of the search tree is $O(n)$, we need only store $O(n)$ IDs. In the dynamic program on the path decomposition, we can employ the same idea by using the coloring of the bags as the ID. In this case as well, we only require polynomial space.
\end{proof}

\section{Exact algorithm for Induced Matching}\label{sec:extend}
\looseness -1
In this section, we would discuss an exact-exponential version of the \textsf{Deletion to Induced Matching} problem. We provide an exact-exponential algorithm of running time $O^*(1.5098^n)$ and polynomial space which uses the key ideas from the parameterized version of \textsf{IND}. Although, the running time is weaker compared to the $O^{*}(1.4231^n)$-time algorithm of Xiao et al.~\cite{AIMexact}, the ideas are much simpler. For an improved running time, we would use the monotone local search \cite{monotone} for exact problems and show an $O^*\paren{1.427^{n+o(n)}}$ running time solution. We state the exact-exponential version of the problem, called \textsf{EXTEND}, as follows:\\

\noindent
\fbox{\parbox{\textwidth-\fboxsep}{
\textsf{Deletion to Induced Matching} (EXTEND)\\
\textbf{Input:} A graph $G$ on $n$ vertices and $m$ edges.\\
\textbf{Task:} determine the minimum cardinality subset $S\subseteq V(G)$ such that the maximum size of any component in $G-S$ is exactly 2.
}}\\
\smallskip

Notice that since we get a fixed parameter tractable algorithm of running time of $O^*(1.748^k)$ for \textsf{IND} thus the running time of the exact-exponential problem is bounded by $O(1.748^n)$ where we solve for $n$ instead of $k$. We would  use the stated \obref{contained} (cf \secref{sec:ind}) to find the minimum possible set $S$ such that the task is fulfilled. Note that the observation has a particular property that it doesn't involve increase in the cardinality of a solution set $S'$ if a solution $S$ exists.\\

We approach the problem in the similar manner where the main idea is to exhaust all the vertices $v$ such that $d(v) \ge$ 4. We state the following observation:
\begin{observation}\label{next}
If $\exists u,v$ $\notin  S$ such that $uv \in E(G)$ then we can delete  $\{u,v\}$ where $G$ $\rightarrow $ $G'$ such that n $\rightarrow$ $n-2$.  
\end{observation}
\noindent
Observations \ref{contained} and \ref{next} suggest the following branching rules where for any rooted node $u$ we have $d(u)$ $\geq 4$ :
\begin{brule}\label{newone}
If  $\exists v \in V(G)$ such that $N[v] \subseteq N[u]$. Then, make nodes in the branch tree for the following cases: $u \in S$ or $u,v$ $\notin S$ using \obref{contained}. 
In the second case we delete at least 3 vertices as d(u) $\ge$ 4 along with $uv$ edge using \obref{next}.
The recurrence relation is $T(n)\leq T(n-1)+T(n-5)$ which solves to $1.3247^n$.

\end{brule}

From now on we assume that for every vertex $u$ of degree at least 4, we have that for every vertex $v\in N(u)$, $|N(v)\cap \overline{N[u]}|\geq 1$. If $u\notin S$, then one neighbor $v\in N(u)$ does not belong to $S$. In that case, $N(v)\cap \overline{N[u]}\subseteq S$. This is the basis for the following branching rule:
\begin{brule}\label{newthree}
Let $u$ be a vertex of degree $d\geq 4$ such that $\forall v\in N(u)$  $|N(v)\cap \overline{N[u]}|\geq 1$.
Create $d+1$ branch node: one for the case when $u\in S$ and one for each vertex $v\in N(u)$ such that 
$N(u)\cup N(v)\setminus \{u,v\}\subseteq S$ where at least d vertices are deleted along with
$\{u,v\}$ using \obref{next}.  Then, the recurrence 
relation is $T(n)\leq T(n-1)+d\cdot T(n-d-2)$  whose solution is bounded by $T(n)\leq 1.5098^n$.
\end{brule}
\looseness -1
\noindent
Using the above Branching Rules \ref{newone}-\ref{newthree}, we can reduce any \textsf{EXTEND} instance $(G)$ to $(G')$ such that $G'$ has maximum degree 3 for any vertex $u \in V(G')$. Now, we would prove a theorem which ascertains an algorithm on an \textsf{EXTEND} instance where every vertex of $G$ has bounded degree 3. 
\begin{theorem}\label{newdp}
There exists an algorithm that given an instance G of max degree 3 solves \textsf{EXTEND} in $O^*(3^{\frac{n}{6}})$ time.
\end{theorem}
\begin{proof}
Note that in \lemref{width}, we find a path decomposition of the \textsf{IND} instance $(G,k)$ with width $p$ such that $p$ is bounded by $\frac{2.5k}{6}+2$. But note that using \lemref{bound}, $G$ can have at max $2.5k$ many vertices of degree 3 which is indeed bounded by $n$.
Now, if we run the algorithm in \thmref{dp} for all possible \textsf{IND} instance $(G,k)$ $\forall k \in \bracket{n}$ where the path decomposition of $G$ has pathwidth $p$ bounded by $\paren{\frac{n}{6}+2}$. Note, using this procedure the running time is bounded by $n*3^{(\frac{n}{6}+2)}*n^c$ where c is a constant from \thmref{dp}. Thus, overall the running time for \textsf{EXTEND} is bounded by $O(3^{\frac{n}{6}})$ (i.e. $O(1.2009^n)$) .  
\end{proof}
Now, we state the main theorem with detailed proof referred to \appref{appendix: extend}.
\begin{theorem}\label{newind}
\textsf{EXTEND} can be solved in $O^*(1.5098^n)$ running time and polynomial space.
\end{theorem}

\medskip
\noindent
{\bf Reduction in complexity using monotone local search} In the seminal work, Saket et al.~\cite{monotone} proposed the technique of monotone local search for designing exact exponential-time algorithms for subset
problems. Their main result is that a $c^kn^{O(1)}$ time
algorithm for the extension problem (e.g. parameterized problem) immediately yields a randomized algorithm for finding a solution of any size with running time $O\paren{\paren{2-\frac{1}{c}}^n}$. The algorithm could be derandomized to running time of $O\paren{\paren{2-\frac{1}{c}}^{n+o(n)}}$. In section \secref{sec:ind}, we show a $O^*(1.748^k)$ running time FPT algorithm for \textsf{IND}. Using monotone local search, the running time for \textsf{EXTEND} could be reduced as follows:
\begin{equation*}
    O^*\paren{\paren{2-\frac{1}{1.748}}^{n+o(n)}} = O^*\paren{1.427^{n+o(n)}} 
\end{equation*}
As stated in \cite{monotone}, this exact algorithm could be shown to run in superpolynomial space complexity.
\vspace{-1mm}
\section{Conclusion}\label{sec:con}
In this work, we provide a fixed parameter tractable algorithm to the deletion to induced matching problem. We provide a novel way to combine ideas from branching and path decomposition, which gives an efficient solution to this problem. Although Xiao et al.~\cite{AIMpara} provides a solution with similar running time, our work provides new insight into the problem. We further use our technique to present a much simpler $O^*(1.5098^n)$-time algorithm for \textsf{EXTEND}. We note that using similar ideas one could provide similar results for various exact version of $\ell$-\textsf{COC} (2-\textsf{COC} in our case). We believe that there could be further improvement to the running time for \textsf{IND} if the Branching Rule-\ref{two} (cf \secref{sec:ind}) is improved to a better bound or a path decomposition with smaller pathwidth for graphs with degree at most 4 is devised. The bottleneck could be analysing a matching case involving an induced bipartite subgraph. We leave the improvement for future work. 

%
%
%






\appendix
\section{APPENDIX}\label{appendix: all contents}
\subsection{Proof of \lemref{width}}\label{appendix: pathwidth}
\begin{proof}
By Lemma \ref{bound}, the maximum number of vertices of degree 3 in $G$ is at most $2.5k$. We obtain a graph
$G'$ by first recursively contracting all edges incident on two degree-2 vertices and then recursively contracting
edges incident on one vertex of degree at most 2. If the degree of a \emph{contracted} vertex is 3, we color the 
edge representing the contracted path red. If a vertex of degree 1 gets \emph{merged} to a vertex of degree 3, we 
color the vertex red. Clearly, the maximum degree of any vertex in $G'$ is at most 3 and $|V(G')|\leq 2.5k$. By 
Theorem \ref{fedor}, we can obtain a path decomposition $\mathcal{P}'$ of width at most $p'=\frac{2.5k}{6}$ in polynomial time.

Now we show how to obtain a path decomposition $\mathcal{P}$ of width at most $p'+1$ of $G$ using the 
decomposition $\mathcal{P}'$ in polynomial time. Suppose there is a red edge in bag $X_t$. Let $x,y$ be the 
endpoints of this edge. Without loss of generality, assume that $y$ is introduced after $x$ has been introduced 
and $X_{t'}$ be the bag just before $y$ is introduced. Clearly, $x\in X_{t'}$. Let $P=xu_1u_2\dots u_ay$ be the 
path in $G$ corresponding to the contracted edge is $G'$. We introduce the following bags in order in 
$\mathcal{P}'$ after $X_{t'}$: $X_{u_1}, X_{u_2}, X_{\overline{u_1}}, X_{u_3},X_{\overline{u_2}}\dots$ where 
$X_u$ denotes an introduce bag for $u$ and $X_{\overline{u}}$ denotes a forget node for $u$. 
Here $X_{u_1}=X_{t'}\cup \{u_1\}$. Note that the size of any bag is at most $p'+2$. This \emph{split} of a bag
containing a red edge is applied exhaustively. Note that there can be at most one red edge between two vertices
that also share an edge in the original graph. If there is a red edge between two vertices which are not adjacent
in the original graph, then in the above sequence when we introduce extra bags, the size of the bags does not 
exceed the original size. Now we consider the case when there are two $xy$-paths in $G$ that get contracted. 
If $xy\in E(G)$, then pathwidth of $G$ is at most 3 and hence the decomposition satisfies the requirements of the
lemma. Otherwise, we follow above procedure of introducing extra bags for the second path before $y$ is introduced.

For red vertices, the procedure is similar. Let $v$ be a red vertex in some bag $X_t$ and let $au_1u_2\dots v$ be
the path that got contracted to $v$. Let $X'$ be the bag just before $v$ is introduced. We insert the following bags
in the decomposition after $X': X_a, X_{u_1}, X_{u_2},X_{\overline{u_1}}\dots$. We apply this operation 
exhaustively for each red vertex. The width of the decomposition is at most $p'+2$. It is easily seen that
$\mathcal{P}$ is a valid path decomposition of $G$.
\end{proof}

\subsection{Proof of \thmref{newind}}\label{appendix: extend}
\begin{proof}
Algorithm solves \textsf{EXTEND} in two phases. In phase one, it applies some reduction rules to 
eliminate degree four vertices.  In phase two, it solves \textsf{EXTEND} on graphs of maximum degree 3. 

After exhaustively applying  Branching rules 
\ref{newone} and \ref{newthree}, we can assume that the maximum degree of a vertex in the 
remaining graph is 3. Hence, this instance can be seen as an instance 
of \textsf{EXTEND} $(G',n - n')$ with maximum degree three. At this point we run the algorithm of Theorem \ref{newdp} to solve \textsf{EXTEND}
optimally. 
Now we proceed to the running time analysis of the algorithm. First, we write the recurrence relations for each of 
the branch rules we have used so far. 
\begin{align*}
    \textnormal{Branch rule}~\ref{newone}\quad& T(n)\leq T(n-1)+T(n-5) &\hspace{-10mm} T(n)\leq 1.3247^n\\
    \textnormal{Branch rule}~\ref{newthree}\quad& T(n)\leq T(n-1)+d\cdot T(n-d-2) &\hspace{-2mm} T(n)\leq 1.5098^n
\end{align*}
Consider the branch tree at the end of the branch phase. Let $G'$ be the instance with which the algorithm of 
Theorem \ref{newdp} is called. Let $|V(G')| = n'$. It is easily shown by induction on $n$ that the number of 
leaves with $|V(G')| = n'$ is bounded by the worst case branching in the branch tree which is $1.5098^{n-n'}$: 
indeed in the base  case $n=1$ is trivial. Consider that the statement is true for any value less than $n$. 
Let $T$ represent the branch tree. Then consider all the nodes with $|V(G')| = n'$ in $T$. If there is a path on
which there is no such node then, delete the subtree rooted at the node that has an edge to the path that leads 
to a node with $|V(G')| = n'$. Now delete all subtrees rooted at nodes with $|V(G')| = n'$ except the root. In this
truncated tree $T^*$, leaves have $G'$ where $|V(G')| = n'$. Now, change $n$ to $n_0=n-n'$. Due to this  change, all 
leaves are with $|V(G')| = n'$. Since, $n_0<n$, by induction, the number of nodes with $|V(G')| = 1$ is bounded
by $1.5098^{n_0}=1.5098^{n-n'}$. Hence, the total running time of the algorithm is bounded by \\
\begin{equation*}
    \sum_{s=0}^{n}1.5098^{n-s}\times 1.2009^{n-s}\leq \sum\limits_{s=0}^{n}1.5098^s\times 1.5098^{n-s}\leq O(1.5098^n)
\end{equation*}
\end{proof}
Since the algorithm follows similar steps as presented for \textsf{IND} (cf \secref{ind}), space complexity of the algorithm for \textsf{EXTEND} could be bounded by using similar approach of defining ID for each node in the search tree and the coloring of the bags for path decomposition. All these operations require at most polynomial space, thus the space complexity is polynomial for \textsf{EXTEND}. 

\end{document}